\documentclass[twocolumn]{ceurart}

\usepackage{graphicx}
\usepackage{subfigure}

\usepackage{amsmath,amsthm,amssymb}
\usepackage{thmtools,thm-restate}
\usepackage{xurl}
\usepackage{hyperref}
\hypersetup{
	colorlinks=true,
	linkcolor=blue,
	citecolor=blue,
	filecolor=blue,
	urlcolor=blue
}
\usepackage[nameinlink]{cleveref}
\crefname{table}{Table}{Tables}

\usepackage{tikz}
\usetikzlibrary{calc,arrows.meta}

\usepackage{enumitem}
\setlist[itemize]{leftmargin=4mm}

\newtheorem{theorem}{Theorem}
\newtheorem{definition}[theorem]{Definition}
\newtheorem{lemma}[theorem]{Lemma}

\newcommand{\sq}[1]{\left [ #1 \right ]}
\newcommand{\prn}[1]{\left ( #1 \right )}

\newcommand{\CSSCL}{CSSCL}
\newcommand{\CAMB}{CAMB}

\newcommand{\eps}{\varepsilon}
\newcommand{\cA}{\mathcal{A}}
\DeclareMathOperator{\DLap}{\mathrm{DLap}}
\DeclareMathOperator{\RMSRE}{\mathrm{RMSRE}}
\DeclareMathOperator*{\E}{\mathbb{E}}
\DeclareMathOperator{\Var}{\mathsf{Var}}

\newcommand{\hx}{\widehat{x}}
\newcommand{\zup}{z^{\uparrow}}
\newcommand{\zUp}{z^{\Uparrow}}
\newcommand{\zdown}{z^{\downarrow}}
\newcommand{\zDown}{z^{\Downarrow}}
\newcommand{\var}{\mathsf{var}}
\newcommand{\varup}{\mathsf{var}^{\uparrow}}
\newcommand{\varUp}{\mathsf{var}^{\Uparrow}}
\newcommand{\vardown}{\mathsf{var}^{\downarrow}}
\newcommand{\varDown}{\mathsf{var}^{\Downarrow}}
\newcommand{\hvar}{\widehat{\var}}

\newcommand{\child}{\mathrm{child}}
\newcommand{\parent}{\mathrm{parent}}
\newcommand{\anc}{\mathrm{anc}}
\newcommand{\N}{\mathbb{N}}
\newcommand{\R}{\mathbb{R}}
\usepackage{algorithm}
\usepackage[noend]{algorithmic}

\definecolor{Gred}{RGB}{219, 50, 54}
\definecolor{Ggreen}{RGB}{60, 186, 84}
\definecolor{Gblue}{RGB}{72, 133, 237}
\definecolor{Gyellow}{RGB}{247, 178, 16}
\definecolor{ToCgreen}{RGB}{0, 128, 0}
\definecolor{myGold}{RGB}{231,141,20}
\definecolor{myBlue}{rgb}{0.19,0.41,.65}
\definecolor{myPurple}{RGB}{175,0,124}

\newcommand{\impCol}[1]{\textcolor{black!30!Gred}{#1}}
\newcommand{\campaignCol}[1]{\textcolor{black!30!Gred}{#1}}
\newcommand{\locationCol}[1]{\textcolor{black!30!Gred}{#1}}
\newcommand{\convCol}[1]{\textcolor{black!30!Ggreen}{#1}}
\newcommand{\histCol}[1]{\textcolor{black!30!Gyellow}{#1}}

\allowdisplaybreaks
\begin{document}

\copyrightyear{2023}
\copyrightclause{Copyright for this paper by its authors.
  Use permitted under Creative Commons License Attribution 4.0
  International (CC BY 4.0).}

\conference{AdKDD'23}

\title{Optimizing Hierarchical Queries for the \\ Attribution Reporting API}

\iffalse
\author{Matthew Dawson, Badih Ghazi, Pritish Kamath, Kapil Kumar, Ravi Kumar, Bo Luan \\ Pasin Manurangsi, Nishanth Mundru, Harikesh Nair, Adam Sealfon, Shengyu Zhu}
\affiliation{
\institution{Google}
\country{}
}
\email{{mwdawson, pritishk, kkapil, luanbo, pasin, nmundru, hsnair, adamsealfon, shengyuzhu}@google.com}
\email{{badihghazi, ravi.k53}@gmail.com}
\else
\author[1]{Matthew Dawson}[email=mwdawson@google.com]
\author[1]{Badih Ghazi}[email=badihghazi@gmail.com]
\cormark[1]
\author[1]{Pritish Kamath}[email=pritishk@google.com]
\author[1]{Kapil Kumar}[email=kkapil@google.com]
\author[1]{Ravi Kumar}[email=ravi.k53@gmail.com]
\author[1]{Bo Luan}[email=luanbo@google.com]
\author[1]{Pasin Manurangsi}[email=pasin@google.com]
\author[1]{Nishanth Mundru}[email=nmundru@google.com]
\author[1]{Harikesh Nair}[email=hsnair@google.com]
\author[1]{Adam Sealfon}[email=adamsealfon@google.com]
\author[1]{Shengyu Zhu}[email=shengyuzhu@google.com]
\address[1]{Google}
\fi

\begin{abstract}
We study the task of performing hierarchical queries based on summary reports from the {\em Attribution Reporting API} for ad conversion measurement. We demonstrate that  methods from optimization and  differential privacy can help cope with the noise introduced by privacy guardrails in the API. In particular, we present algorithms for (i) denoising the API outputs and ensuring consistency across different levels of the tree, and (ii) optimizing the privacy budget across different levels of the tree. We provide an experimental evaluation of the proposed algorithms on public datasets.
\end{abstract}

\begin{keywords}
ad conversion measurement \sep
hierarchical aggregation \sep
differential privacy \sep
attribution reporting API
\end{keywords}

\maketitle

\section{Introduction}

Over the last two decades, third-party cookies \cite{wiki:HTTP_cookie} have been essential to online advertising, and particularly to ad conversion measurement, whereby an ad impression (e.g., a click or a view) on a publisher site or app could be joined to a conversion on the advertiser, in order to compute aggregate conversion reports (e.g., the number of conversions attributed to a subset of impressions) or to train ad bidding models (e.g., \cite{lu2017practical, choi2020delayed, qiu2020predicting, gu2021estimating}). However, in recent years, privacy concerns have led several browsers to decide to deprecate third-party cookies, e.g., \cite{safari, mozilla, chromium}. The \emph{Attribution Reporting API} \cite{chrome-attribution-reporting, aggregate-api-android} seeks to provide privacy-preserving ways for measuring ad conversions on the Chrome browser and the Android mobile operating system. This API relies on a variety of mechanisms for limiting the privacy leakage, including bounding the contributions to the output reports of the conversions attributed to each impression, as well as noise injection to satisfy differential privacy (for more details, see \Cref{sec:preliminaries}).

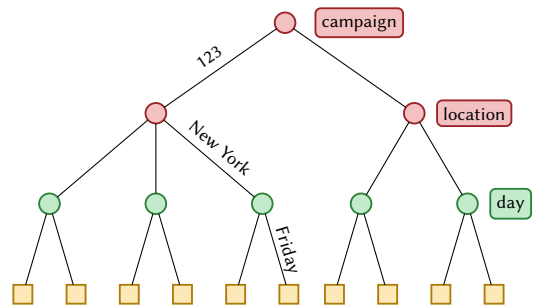
\begin{figure}[thb]
\centering
\newcommand{\crl}[1]{\{ #1 \}}
\begin{tikzpicture}[
  grow=down,
  level 1/.style={sibling distance=34mm},
  level 2/.style={sibling distance=14mm},
  level 3/.style={sibling distance=7mm},
  level 4/.style={sibling distance=5mm},
  level distance=12mm,
  vert/.style={circle, minimum size=2.5mm, draw, line width=1pt},
  leaf/.style={rectangle, minimum size=2.5mm, draw=black!30!Gyellow, line width=0.7pt, fill=Gyellow!30},
  campaign/.style={draw=black!30!Gred, fill=Gred!30, line width=0.7pt},
  geo/.style={draw=black!30!Gred, fill=Gred!30, line width=0.7pt},
  conv/.style={draw=black!30!Ggreen, fill=Ggreen!30, line width=0.7pt},
  textbox/.style={rectangle, rounded corners=2pt, line width=1pt},
  every node/.style = {outer sep=0pt},
  scale=1, transform shape
]

\node[vert, campaign] (1) {}
  child {
    node[vert, geo] (2) {}
    child {
      node[vert, conv] (4) {}
      child {
        node[leaf] (9) {}
      }
      child {
        node[leaf] (10) {}
      }
    }
    child {
      node[vert, conv] (5)  {}
      child {
        node[leaf] (11) {}
      }
      child {
        node[leaf] (12) {}
      }
    }
    child {
      node[vert, conv] (6) {}
      child {
        node[leaf] (13) {}
      }
      child {
        node[leaf] (14) {}
      }
    }
  }
  child {
    node[vert, geo] (3) {}
    child {
      node[vert, conv] (7) {}
      child {
        node[leaf] (15) {}
      }
      child {
        node[leaf] (16) {}
      }
    }
    child {
      node[vert, conv] (8) {}
      child {
        node[leaf] (17) {}
      }
      child {
        node[leaf] (18) {}
      }
    }
  };

\fontsize{7}{15}\selectfont
\node[textbox, campaign, right=4mm] at (1) {campaign};
\node[textbox, geo, right=3mm] at (3) {location};
\node[textbox, conv, right=3mm] at (8) {day};

\node[rotate=34] at ( $(1)!0.5!(2)+(-0.15,0.15)$ ) {123};
\node[rotate=-40] at ( $(2)!0.5!(6)+(0.15,0.15)$ ) {New York};
\node[rotate=-73] at ( $(6)!0.5!(14)+(0.15,-0.02)$ ) {Friday};
\end{tikzpicture}
\caption{Example of Hierarchical Queries.}
\label{fig:hierarchical_queries}
\end{figure}

We study the \emph{conversion reporting} task, where a query consists of counting the number of conversions attributed to impressions such that some features of the conversion and the impression are restricted to certain given values.  In particular, we focus on the \emph{hierarchical queries} setting where the goal is to estimate the number of conversions attributed to impressions where the features are restricted according to certain nested conditions.  Consider the example in \Cref{fig:hierarchical_queries}, where we wish to estimate the number of conversions:
\begin{itemize}[nosep]
\item attributed to impressions from campaign $123$.
\item that are also restricted to take place in New York.
\item that are further restricted to occur on a Friday.
\end{itemize}
In general, the goal is to estimate the conversion count for each node in a given tree, similar to the one in \Cref{fig:hierarchical_queries}.

Such estimates can be obtained using summary reports from the Attribution Reporting API, as discussed in \Cref{sec:agg_api_summary}. In this work, we present a linear-time post-processing algorithm that denoises the estimates for different nodes that are returned by the API and ensures that the estimates are consistent with respect to the tree structure. We also show that our  algorithm is optimal among all linear unbiased estimators for arbitrary trees, extending results for regular trees~\cite{hay2009boosting, cormode2012differentially} (\Cref{sec:post-processing}).  Since the API allows the ad-tech to allocate a privacy budget across different measurements containing contributions from the same impression, we provide an algorithm for optimizing the allocation of the privacy budget across the different levels of the tree (\Cref{sec:privacy_budgeting}).

We start by recalling in \Cref{sec:preliminaries} the basics of differential privacy and summarizing the query model for summary reports in the Attribution Reporting API. In \Cref{sec:optimization_problem}, we formally define the optimization problem that we seek to solve. In \Cref{sec:experimental_evaluation}, we provide an experimental evaluation of our algorithms on two public datasets. We conclude with some future directions in \Cref{sec:conclusion}.

\section{Preliminaries}\label{sec:preliminaries}

\subsection{Differential Privacy (DP)}
Let $n$ be the number of rows in the dataset and let $\mathcal{X}$ be the (arbitrary) set representing the domain of values for each row. We  distinguish two types of columns (a.k.a. attributes): {\em known} and {\em unknown}. We also assume knowledge of the set of possible values that each unknown attribute can take.

\begin{definition}[DP \cite{DworkMNS06}]\label{def:differential_privacy}
For $\eps \geq 0$, an algorithm $\cA$ is \emph{$\eps$-DP} if for every pair $X, X'$ of datasets that differ on the unknown attributes of one row\footnote{We note that this instantiation of the DP definition is related to the label DP setting in machine learning \cite{chaudhuri2011differentially, ghazi2021deep, esfandiari2022label, GhaziKKLMVZ2022}, where the features of an example are considered known and only its label is deemed unknown and sensitive. In our use case, there might be multiple unknown attributes, whose concatenation is treated in a conceptually similar way to the {\em label} in the label DP setting.}, and for every possible output $o$, it holds that $\Pr[\cA(X) = o] \leq e^\eps \cdot \Pr[\cA(X') = o]$.
\end{definition}

\begin{lemma}[Basic Composition]\label{le:basic_composition_DP}
Let $\cA$ be an algorithm that runs $k$ algorithms $\mathcal{A}_1, \dots, \mathcal{A}_{k}$ on the same dataset such that $\mathcal{A}_i$ is $\eps_i$-DP with $\eps_i \geq 0$ for each $i \in [k]$. Then, $\mathcal{A}$ is $(\sum_{i=1}^k \eps_i)$-DP.
\end{lemma}

\begin{lemma}[Post-processing]\label{le:post-processing}
    Let $\eps > 0$, and $R$ and $R'$ be any two sets. If $\mathcal{A}: \mathcal{X}^n \to R$ is an $\eps$-DP algorithm, and $f: R \to R'$ is any randomized mapping, then $(f \circ \mathcal{A}): \mathcal{X}^n \to R'$ is $\eps$-DP. 
\end{lemma}

For an extensive overview of DP, we refer the reader to the monograph \cite{dwork2014algorithmic}. A commonly used method in DP is the discrete Laplace mechanism. To define it, we recall the notion of $\ell_1$-sensitivity.
\begin{definition}[$\ell_1$-sensitivity]
    Let $\mathcal{X}$ be any set, and $f: \mathcal{X}^n \to \mathbb{R}^d$ be a $d$-dimensional function. Its \emph{$\ell_1$-sensitivity} is defined as $\Delta_1 f := \max_{X, X'} \|f(X) - f(X')\|_1$, where $X$ and $X'$ are two datasets that differ on the unknown attributes of a single row.
\end{definition}

\begin{definition}[Discrete Laplace Mechanism]\label{def:disc_lap_mech}
    The \emph{discrete Laplace distribution} centered at $0$ and with parameter $a > 0$, denoted by $\DLap(a)$, is the distribution whose probability mass function at integer $k$ is $\frac{e^{a} - 1}{e^{a} + 1} \cdot e^{- a|k|}$.
    The \emph{$d$-dimensional discrete Laplace mechanism} with parameter $a$ applied to a function $f: \mathcal{X}^n \to \mathbb{Z}^d$, on input a dataset $X \in \mathcal{X}^n$, returns $f(X) + Z$ where $Z$ is a $d$-dimensional noise random variable whose coordinates are sampled i.i.d. from $\DLap(a)$.
\end{definition}

\begin{lemma}\label{lem:DP_Discrete_Laplace_Mechanism}
For every $\eps > 0$, the $d$-dimensional discrete Laplace mechanism with parameter 
$a \le \eps / \Delta_1 f$
is $\eps$-DP.
\end{lemma}

\subsection{Attribution Reporting API}\label{sec:agg_api_summary}

The \emph{aggregatable reports} \cite{aggregate-api-android} are constructed as follows:
\begin{itemize}[leftmargin=*]
\item {\bf Impression (a.k.a. source) registration:} The API provides a mechanism for the ad-tech to register an impression (e.g., a click or view) on the publisher site or app.  During registration, the ad-tech can specify impression-side aggregation keys (e.g., one corresponding to the campaign or geo location).

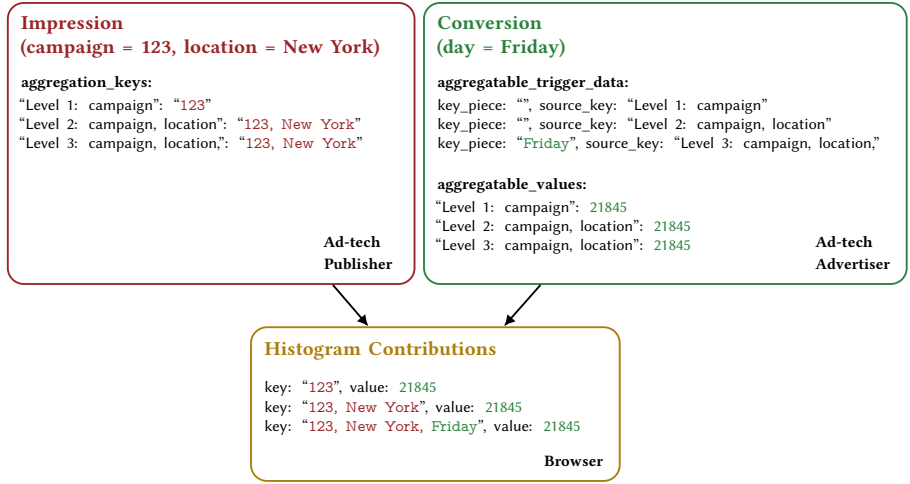
\begin{figure*}[t]
\centering
\begin{tikzpicture}
\tikzset{
	box/.style = {rectangle, rounded corners=6pt, line width=0.7pt, draw, inner sep=5pt}
}
\node[box, text width=5cm, text depth=3.2cm, below right, draw=black!30!Gred] (imp) at (0, 0) {
    {\footnotesize
    \impCol{\bf Impression}\\
    \impCol{\bf (campaign = 123, location = New York)}\\[1mm]
    }
    
    {\bf \scriptsize aggregation\_keys:}\\[0.5mm]
    
    {\scriptsize
    ``Level 1: campaign'': ``\impCol{\texttt{123}}'' \\
    ``Level 2: campaign, location'': ``\impCol{\texttt{123, New York}}'' \\
    ``Level 3: campaign, location,'': ``\impCol{\texttt{123, New York}}''
    \\[2.5mm]
    } 
};

\node[text width=8mm, above left] at ($(imp.south east)+(-0.3,0.1)$){\bf \scriptsize
	Ad-tech\\[-0.5mm]
	Publisher
};

\node[box, text width=6cm, text depth=3.2cm, below right, draw=black!30!Ggreen] (conv) at ($(imp.north east) + (0.1,0)$) {
    {\footnotesize
    \convCol{\bf Conversion}\\
    \convCol{\bf (day = Friday)}\\[1mm]
    }
    
    {\bf \scriptsize aggregatable\_trigger\_data:}\\[0.5mm]
    {\scriptsize
    key\_piece:  ``'', source\_key: ``Level 1: campaign''\\
    key\_piece: ``'', source\_key: ``Level 2: campaign, location''\\
    key\_piece: ``\convCol{\texttt{Friday}}'', source\_key: ``Level 3: campaign, location,''\\[2mm]
    }
    {\bf \scriptsize aggregatable\_values:}\\[0.5mm]
    {\scriptsize
    ``Level 1: campaign'': \convCol{21845}\\
    ``Level 2: campaign, location'': \convCol{21845}\\
    ``Level 3: campaign, location'': \convCol{21845} \\[2.5mm]
    }
};
\node[text width=8mm, above left] at ($(conv.south east)+(-0.3,0.1)$){\bf \scriptsize
    Ad-tech\\[-0.5mm]
    Advertiser
};

\node[box, text width=4.5cm, text depth=1.5cm, below right, draw=black!30!Gyellow] (hist) at (3.2, -4.3) {
    {\footnotesize
    \histCol{\bf Histogram Contributions}\\[2.5mm]
    }
    
    {\scriptsize
    key: ``\impCol{\texttt{123}}'', value: \convCol{21845}\\
    key:  ``\impCol{\texttt{123, New York}}'', value: \convCol{21845}\\
    key: ``\texttt{\impCol{123, New York,} \convCol{Friday}}'', value: \convCol{21845}\\[2.5mm]
    }
};
\node[text width=6mm, above left] at ($(hist.south east)+(-0.3,0.1)$){\bf \scriptsize
	Browser
};

\path[-{Latex[width=1.5mm,length=1.5mm]}, line width=0.7pt]
(imp) edge (hist)
(conv) edge (hist);
\end{tikzpicture}
\caption{Histogram Contributions Generation.}
\label{fig:hist_contrib_construct}
\end{figure*}

\item {\bf Conversion (a.k.a. trigger) registration:} The API also provides a mechanism for the ad-tech to register a conversion on the advertiser site or app. As it does so, the ad-tech can specify conversion-side key pieces along with aggregatable values corresponding to each setting of the impression-side aggregation keys. E.g., a conversion-side key piece could capture the conversion type or (a discretization of) the conversion timestamp. The combined aggregation key (which can be thought of as the concatenation of the impression-side aggregation key and the conversion-side key piece) is restricted to be at most $128$ bits. The aggregatable value is required to be an integer between $1$ and the \emph{$L_1$ parameter} of the API, which is set to $2^{16} = 65{,}536$.
\item {\bf Attribution:} The API supports last-touch attribution where the conversion is attributed to the last (unexpired) impression registered by the same ad-tech. (The API supports a broader family of single-touch attribution schemes allowing more flexible prioritization over impressions and conversions; we do not discuss this aspect any further as it is orthogonal to our algorithms.\footnote{The API moreover enforces a set of rate limits on registered impressions, and attributed conversions; we omit discussing these since they are not essential to the focus of this paper. We refer the interested reader to the  documentation \cite{attribution-reporting-wicg}.})
\item {\bf Histogram contributions generation:} The API enforces that the sum of contributions across all aggregatable reports generated by different conversions attributed to the same impression is capped to at most the $L_1 = 2^{16}$ parameter. An example construction of histogram contributions is in \Cref{fig:hist_contrib_construct}. While in this example the conversion key piece depends only on conversion information (the conversion day), it can be replaced by an attribute that depends on both impression and conversion information (e.g., a discretization of the difference between impression time and conversion time); this can be done using filters \cite{agg-api-trigger-filters}. 
\end{itemize}

At query time, a set of histogram keys is requested by the ad-tech \cite{aggregation-service}. (The set of keys that are queried could be set to the Cartesian product of all \emph{known} values of the impression-side features with the set of all \emph{possible} values of the conversion-side features.) The aggregatable reports are combined in the aggregation service to produce a \emph{summary report} by applying the discrete Laplace mechanism (\Cref{def:disc_lap_mech}) with parameter $\eps/L_1$ to the requested aggregation keys. This report satisfies $\eps$-DP where each row of the dataset corresponds to an impression and its attributed conversions if any (see \Cref{fig:dataset_for_DP} for an example\footnote{The 
$^*$ in the conversion-related field in \Cref{fig:dataset_for_DP} indicates that the click corresponding to that row did not get an attributed conversion.}), and where the known columns in \Cref{def:differential_privacy} are the attributes that  only depend on impression information (campaign and location in \Cref{fig:dataset_for_DP}).

{\fontsize{8.5}{12}\selectfont
\renewcommand{\arraystretch}{1.3}
\begin{table}
\caption{Dataset Post-Attribution and Pre-Aggregation.}
\begin{tabular}{c|l|l|l}
\hline
{\bf Click} & \campaignCol{\bf campaign} & \locationCol{\bf location} & \convCol{\bf day}\\
\hline
1 & \campaignCol{123} & \locationCol{Paris} & \convCol{Monday}\\
2 & \campaignCol{456} & \locationCol{Chicago} & \convCol{Friday}\\
3 & \campaignCol{789} & \locationCol{London} & \convCol{*}\\
4 & \campaignCol{123} & \locationCol{New York} & \convCol{Friday}\\
$\cdots$ & \campaignCol{$\cdots$} & \locationCol{$\cdots$} & \convCol{$\cdots$}\\
\hline
\end{tabular}
\label{fig:dataset_for_DP}
\end{table}}
In the hierarchical aggregation setting, each node of the tree corresponds to an aggregation key, and the level of the node is (implicitly) specified in the impression-side aggregation key and/or in the conversion-side key piece (as shown in \Cref{fig:hist_contrib_construct}).

For the use case of estimating conversion counts, the aggregatable value could be set so as to increment the count by $+1$. Since the scale of the noise injected in summary reports is $L_1/\eps$, the ad-tech can improve accuracy by setting the contribution of an increment to $+L_1$ instead of $+1$ (and then scaling down the value it receives from the aggregation service by $L_1$). If the $L_1$ contribution has to be divided across multiple keys, the contribution of each increment needs to be scaled down accordingly. E.g., in \Cref{fig:hist_contrib_construct}, since each impression affects $3$ keys, the contribution is set to $\lfloor L_1 / 3 \rfloor = 21,845$.

\section{Optimization Problem}\label{sec:optimization_problem}

\subsection{Hierarchical Query Estimation}

\noindent We formally define the {\em hierarchical query estimation} problem. Given a dataset $X$, consider a tree where each node corresponds to the subset of the rows of $X$, conditioned on the values of some of the attributes. We consider the setting where each level of the tree introduces a conditioning on the value of a new attribute. For {\em known} attributes, the child nodes of a node correspond to the different values taken by that attribute in $X$ within the rows. For {\em unknown} attributes, the child nodes correspond to {\em all} possible values for that attribute, whether or not they actually occur in the dataset. Given this tree, the problem is to privately release the approximate number of data rows corresponding to each node that have an attributed conversion.

\subsection{Error Measure and Consistency}
We consider the following error measure, which is defined in \cite{noise-lab}.

\begin{definition}[RMS Relative Error at Threshold]
For a count $c \geq 0$, and its randomized estimate $\hat{c} \in \mathbb{R}$, the \emph{Root Mean Squared Relative Error at Threshold $\tau$} when estimating $c$ by $\hat{c}$ is defined as
$
\RMSRE_\tau(c, \hat{c}) := \sqrt{\E\sq{\prn{\frac{|\hat{c} - c|}{\max(\tau, c)}}^2}}\,,
$
where the expectation is over the randomness of $\hat{c}$.    
\end{definition}

Suppose we have the count estimates (e.g., the number of conversions as 
in \Cref{fig:hierarchical_queries}) at every node in a tree.  Each conversion contributes to the count for multiple nodes. For example, a conversion for ad campaign 123 that occurs on a Friday in New York contributes to the $6$th leaf from the left, but also to each of its ancestor nodes. This imposes relationships among the counts at various nodes. If the only geo locations with conversions attributed to ad campaign 123 on Friday are New York and Chicago, then the total number of such conversions must equal the sum of the number of such conversions in each of the two locations. More generally, the count for any node must equal the sum of the counts for its children. An estimator with this property is called \emph{consistent}.

For a tree $T$ with levels $L_0$, $\ldots$, $L_d$, with $L_i$ being the set of nodes at level $i$, and estimators
$\hat{c}_v$ of the counts $c_v$ at each node $v$, define 
the tree error $\RMSRE_\tau(T)$ to be
\[
\sqrt{\frac{1}{d+1} \sum_{i=0}^d 
\left(
\frac{1}{|L_i|} \sum_{v \in L_i} \RMSRE_\tau(c_v, \hat{c}_v)^2 \right)}\,.
\]

The goal of the hierarchical query estimation problem is to privately estimate the counts of every node with minimum possible tree error, where the estimates should be consistent.  To achieve this, we will employ post-processing and privacy budgeting strategies. 

\section{Post-processing Algorithms}\label{sec:post-processing}

Directly applying the discrete Laplace mechanism to add independent noise to each node does not result in a consistent estimate.
Consistency can be achieved by estimating only the counts of leaves  and inferring the count of each
nonleaf node by adding the counts of its leaf descendants, but this can lead to large error for nodes higher up in the tree.
Alternatively, one can achieve consistency by post-processing independent per-node estimates. Since DP is preserved under post-processing (\Cref{le:post-processing}), this comes at no cost to the privacy guarantee, and it can substantially improve accuracy. 

Since the count of any node must equal the sum of the counts of its children, we can obtain a second independent estimate of the count of any nonleaf node by summing the estimates of its children. 
We can combine these two estimates to obtain a single estimate of lower variance.
Extending this observation, Hay et al. \cite{hay2009boosting} and Cormode et al. \cite{cormode2012differentially} give efficient post-processing algorithms for regular (every non-leaf has the same number of children) and balanced (every leaf is at the same depth) trees, achieving consistency and also a substantial improvement in accuracy. In particular, estimating the counts of each node can be expressed as a linear regression problem, and these algorithms compute the least-squares solution, which is known to achieve optimal error variance among all unbiased linear estimators.  Moreover, this special case of least-squares regression can be solved in linear time.

\newcommand{\CombineEstimates}{\mathsf{CombineEstimates}}
\begin{algorithm}[t]
\caption{$\CombineEstimates$}
\label{alg:combine-estimates}
\begin{algorithmic}
\STATE {\bf Input:} $(x; \var_x)$,  $(y; \var_y)$ : $x$, $y$ are samples from random variables $X, Y$ resp. such that $\E X = \E Y$ and variances $\Var(X) = \var_x$ and $\Var(Y) = \var_y$.
\STATE {\bf Output:} $(z; \var_z)$ : combined estimate and variance.
\STATE %
\STATE $\var_z \gets \frac{\var_x \cdot \var_y}{\var_x + \var_y}$
\STATE $z \gets \var_z \cdot \prn{\frac{x}{\var_x} + \frac{y}{\var_y}}$
\RETURN $(z; \var_z)$
\end{algorithmic}
\end{algorithm}

\begin{algorithm*}[t]
\caption{$\mathsf{TreePostProcessing}$} %
\label{alg:post-processing}
\begin{algorithmic}
\STATE {\bf Params:} Tree $T$ with root $r$. %
\STATE {\bf Input:} $(x_v; \var_v)_{v \in T}$: noisy counts and variances for all $v \in T$.%
\STATE {\bf Output:} $(\hx_v; \hvar_v)_{v \in T}$: post-processed estimates and variance for all $v$.
\STATE %
\STATE \textcolor{black!60}{\emph{\# Bottom-up pass}}
\FOR{leaf $v\in T$}
    \STATE $\triangleright\ (\zUp_v; \varUp_v) \gets (x_v; \var_v)$
\ENDFOR
\FOR{each internal node $v$ from largest to smallest depth}
    \STATE $\triangleright\ (\zup_v; \varup_v) \gets \left(\sum_{u\in\child(v)}\zUp_u;  \sum_{u\in\child(v)} \varUp_u \right)$
    \STATE $\triangleright\ (\zUp_v; \varUp_v) \gets \CombineEstimates((x_v; \var_v), (\zup_v; \varup_v))$
\ENDFOR
\STATE %
\STATE \textcolor{black!60}{\emph{\# Top-down pass}}
\STATE For root $r$:
\STATE $\triangleright\ (\hx_r; \hvar_r) \gets (\zUp_r; \varUp_r)$
\STATE $\triangleright\ (\zDown_r; \varDown_r) \gets (x_r; \var_r)$
\FOR{each non-root node $v$ from smallest to largest depth}
    \STATE $\triangleright\ p \gets \parent(v)$
    \STATE $\triangleright\ (\zdown_v; \vardown_v) \gets \prn{\zDown_p - \hspace{-3mm} \sum\limits_{u \in \child(p) \smallsetminus \{v\}} \hspace{-3mm} \zUp_u; \varDown_p + \hspace{-3mm} \sum\limits_{u \in \child(p) \smallsetminus \{v\}} \hspace{-3mm} \varUp_u}$\hfill
    \textcolor{black!60}{\emph{\# equals $(\zDown_p - \zup_p + \zUp_v; \varDown_p + \varup_p - \varUp_v)$}}
    \STATE $\triangleright\ (\hx_v; \hvar_v) \gets \CombineEstimates((\zUp_v; \varUp_v), (\zdown_v; \vardown_v))$
    \STATE $\triangleright\ (\zDown_v; \varDown_v) \gets \CombineEstimates((x_v; \var_v), (\zdown_v; \vardown_v))$
\ENDFOR
\RETURN $(\hx_v; \hvar_v)_{v\in T}$
\end{algorithmic}
\end{algorithm*}

We generalize this algorithm to arbitrary trees in \Cref{alg:post-processing}.
This allows us to handle trees with different fanouts and noise at different levels or even at different nodes in the same level,
which is the case for the conversion reporting trees we study. The input to the algorithm are independent noisy counts $x_v$ for all nodes $v$ in $T$, with $\E x_v = c_v$ and variance $\mathsf{Var}(x_v) = \var_v$.

The key idea is a simple method (\Cref{alg:combine-estimates}) to combine two {\em independent} and {\em unbiased} estimates of the same quantity to get an estimate with reduced variance as follows: Suppose $X$ and $Y$ are two {\em independent} random variables with $\E X = \E Y = C$, with variances $\var_X$ and $\var_Y$ respectively, then the optimal convex combination of $X$ and $Y$ that minimizes the variance in the estimate of $C$ is inversely proportional to the variances, namely $Z = (\var_X \cdot Y + \var_Y \cdot X) / (\var_X + \var_Y)$, which has variance $\Var(Z) = (\var_X \cdot \var_Y) / (\var_X + \var_Y)$.

\Cref{alg:post-processing} comprises of two linear-time passes. The first pass proceeds bottom-up from the leaves to the root.
For each non-leaf node $v$ it recursively computes $\zup_v$, which is an optimal linear unbiased estimate of $c_v$, using only the noisy counts $x_u$ for all $u$ that are in the sub-tree rooted at $v$, {\em excluding} $v$. This is combined with the noisy count $x_v$ to compute $\zUp_v$, which is an optimal linear unbiased estimate of $c_v$, using only the noisy counts $x_u$ for all $u$ that are in the sub-tree rooted at $v$, {\em including} $v$.

The second pass proceeds top-down and computes for each non-root node $v$, a estimate $\zdown_v$, which is an optimal linear unbiased estimate of $c_v$ using only the noisy counts $x_u$ for all $u$ that are {\em not} in the sub-tree rooted at $v$.
The estimates $\zUp_v$ and $\zdown_v$ are then combined to obtain $\hx_v$, which is the optimal linear unbiased estimate of $c_v$ using all the estimates $x_u$ for all $u$ in $T$.
To assist with the recursive procedure, it also computes an estimate $\zDown_v$, which is the optimal linear unbiased estimate of $c_v$ using only the noisy counts $x_u$ for all $u$ that are not strictly in the sub-tree under $v$ (i.e., includes $x_v$). See \Cref{fig:post-processing}.

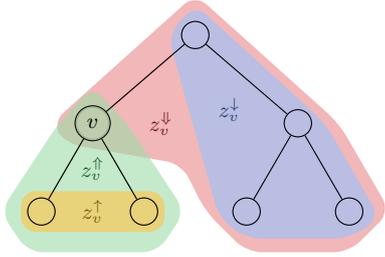
\begin{figure}
\centering
\begin{tikzpicture}[
  vert/.style = {draw, circle, minimum size=4mm},
  level distance=1.3cm,
  level 1/.style={sibling distance=3cm},
  level 2/.style={sibling distance=1.5cm},
  scale = 0.9, transform shape
]
    \def\drawtree{
    \node[vert] (A) {}
    child {
        node[vert] (B) {$v$}
        child { node[vert] (D) {} }
        child { node[vert] (E) {} }
    }
    child {
        node[vert] (C) {}
        child { node[vert] (F) {} }
        child { node[vert] (G) {} }
    };
    }
    \drawtree

    \draw[fill=Gred!50, draw=none, rounded corners=2mm, inner sep=3cm, fill opacity=0.7]
        ($(B)+(-0.6,0)$) --
        ($(A)+(0,0.6)$) --
        ($(C)+(0.6,0)$) --
        ($(G)+(0.6,0)$) --
        ($(G)+(0,-0.6)$) --
        ($(F)+(0,-0.6)$) --
        ($(F)+(-0.6,0)$) --
        ($(A)+(-0.2,-1.9)$) --
        ($(B)+(-0.4,-0.3)$) -- cycle;
    \draw[fill=Gblue!50, draw=none, rounded corners=2mm, inner sep=3cm, fill opacity=0.7]
        ($(A)+(0,0.45)$) --
        ($(C)+(0.4,0)$) --
        ($(G)+(0.45,-0.1)$) --
        ($(G)+(0,-0.4)$) --
        ($(F)+(0,-0.4)$) --
        ($(F)+(-0.4,-0.1)$) --
        ($(A)+(-0.4,0)$) -- cycle;
    \draw[fill=Ggreen!50, draw=none, rounded corners=2mm, inner sep=3cm, fill opacity=0.6]
        ($(B)+(0,0.55)$) --
        ($(B)+(0.5,0)$) --
        ($(E)+(0.6,0)$) --
        ($(E)+(0.2,-0.6)$) --
        ($(D)+(-0.2,-0.6)$) --
        ($(D)+(-0.6,0)$) --
        ($(B)+(-0.5,0)$) -- cycle;
    \draw[fill=Gyellow!70, draw=none, rounded corners=2mm, inner sep=3cm, fill opacity=0.7]
        ($(D)+(-0.3,0.3)$) --
        ($(E)+(0.3,0.3)$) --
        ($(E)+(0.3,-0.3)$) --
        ($(D)+(-0.3,-0.3)$) --
        cycle;
    \normalsize
    \node[black!50!Gred] at ($(B)+(1,0)$) {$\zDown_v$};
    \node[black!50!Gblue] at ($(C)+(-1,0.2)$) {$\zdown_v$};
    \node[black!50!Ggreen] at ($(B)+(0,-0.67)$) {$\zUp_v$};
    \node[black!50!Gyellow] at ($(B)+(0,-1.3)$) {$\zup_v$};

    \drawtree
\end{tikzpicture}
\caption{Intermediate variables in \Cref{alg:post-processing}.}
\label{fig:post-processing}
\end{figure}

\Cref{alg:post-processing} not only reduces the variance of each estimate but also achieves the smallest possible variance among all unbiased linear estimators, simultaneously for all nodes.

\begin{restatable}[Optimality of Post-Processing]{theorem}{lsqOptimality}\label{thm:least-squares-optimal}
For every $v \in T$, $\hx_v$ is the best linear unbiased estimator (BLUE) of the count $c_v$ and has variance $\hvar_v$.
In particular, $(\hx_v)_{v \in T}$ minimizes $\RMSRE_\tau(c_v, \hx_v)$ among all linear unbiased estimators, for all $v \in T$.
\end{restatable}

This extends the results from \cite{hay2009boosting,cormode2012differentially}, which work only for regular trees, 
to arbitrary trees. 
Similar to their proofs, the theorem above follows once we show that the (appropriately scaled) estimates $(\hx_v)_{v \in T}$ are an ordinary least squares estimator (OLS) of the counts $(c_v)_{v \in T}$. Proving the latter boils down to showing that the estimates satisfy the two conditions in the following lemma:

\begin{restatable}{lemma}{postProcessProps} \label{lem:post-process-properties}
For any $T$ and $(x_v; \var_v)_{v \in T}$, let $(\hx_v; \hvar_v)_{v \in T}$ be the output of \Cref{alg:post-processing}. Then, the following hold:
\begin{itemize}%
\item (Consistency) For all internal $v$ : $\hx_v = \sum_{u \in \child(v)} \hx_u$.
\item (Weighted Root-to-Leaf Sum Preservation) For each leaf $v$, $\sum_{u \in \anc(v)} \frac{x_u}{\var_u}  = \sum_{u \in \anc(v)} \frac{\hx_u}{\var_u}$,
where $\anc(v)$ denotes the nodes on the path from $v$ to $r$ (inclusive).
\end{itemize}
\end{restatable}

The proof of \Cref{lem:post-process-properties} is by induction on the number of nodes in the tree. The inductive step is done by selecting a node whose children are all leaves and ``coalescing'' all of the node's children. We provide the full proofs in \Cref{apx:least-squares-optimality}.

\section{Privacy Budgeting Over Tree Levels
}\label{sec:privacy_budgeting}

To allocate the privacy budget across the levels of the tree, the simplest approach is to divide it equally among the levels, or to put all of the budget on the lowest level. However, basic composition (\Cref{le:basic_composition_DP}) allows us to allocate the privacy budget arbitrarily among the nodes of the tree, and we can apply post-processing (\Cref{alg:post-processing}) to noisy initial estimates with unequal variances as well. This motivates the question of whether we can improve accuracy with an unequal privacy budget allocation, and if so, how.

Given the true counts $c$, we use a greedy iterative approach for optimizing the privacy budget allocation. Let $k$ be the number of phases (e.g.,\ $k=20$) corresponding to the granularity of the allocation. Initially allocate zero (or infinitesimal) privacy budget to each level, and divide the (remaining) privacy budget into $k$ units of size $\eps/k$. In each of $k$ phases, select the level that would result in the lowest $\RMSRE_\tau(T)$
when using $\eps/k$ additional privacy budget,
and increase the privacy budget of that level by $\eps/k$.
The error $\RMSRE_\tau(T)$ can be computed directly using the variance $\hvar_v$ of each post-processed estimate $\hx_v$ as returned by \Cref{alg:post-processing}. We present the details of this greedy iterative approach in \Cref{apx:greedy-budgeting}.

Using the true counts to optimize the privacy budget allocation can leak sensitive information and violate the privacy guarantee, and is infeasible in the API. Instead of using the true counts to allocate the privacy budget, one can use alternatives such as simulated data, or historical data that  is not subject to the privacy constraints (e.g., before third-party cookie deprecation), or noisy historical data
that has already been protected with DP (e.g., the output of the API over data from a previous time period). We refer to this family of privacy budget optimization methods as \emph{prior-based}. When no such alternatives are available, one could start out with a suboptimal privacy budgeting strategy (e.g.,\ uniform privacy budgeting across levels) and improve the allocation over time.

\section{Experimental Evaluation}\label{sec:experimental_evaluation}

We evaluate the algorithms on two public ad conversion datasets.

\paragraph{Criteo Sponsored Search Conversion Log (\CSSCL) Dataset~\citep{tallis2018reacting}}

This dataset contains $15{,}995{,}634$ clicks obtained from a sample of $90$-day logs of live traffic from Criteo Predictive Search. Each point contains information on a user action (e.g., time of click on an ad) and a potential subsequent conversion (purchase of the corresponding product) within a $30$-day attribution window. We consider the attributes:
\texttt{partner\_id},
\texttt{product\_country},
\texttt{device\_type},
\texttt{product\_age\_group} \&
\texttt{time\_delay\_for\_conversion}.
The last attribute is discretized into $2$-day (or $6$-day) buckets so that there are at most $15$ ($5$ respectively) possible values of the rounded time delay. The first $4$ attributes are considered {\em known} (they only relate to the impression), whereas the last is considered {\em unknown}. For the discretization into $6$-day buckets, we retain the unknown attribute and $3$ known attributes instead of $4$ (omitting \texttt{product\_age\_group}), resulting in a depth $4$ tree (instead of depth $5$).

\newcommand{\myfigwidth}{0.9\columnwidth}
\begin{figure*}
  \centering
  \subfigure[$\tau=5$, depth $5$ tree]{
    \includegraphics[width=\myfigwidth]
    {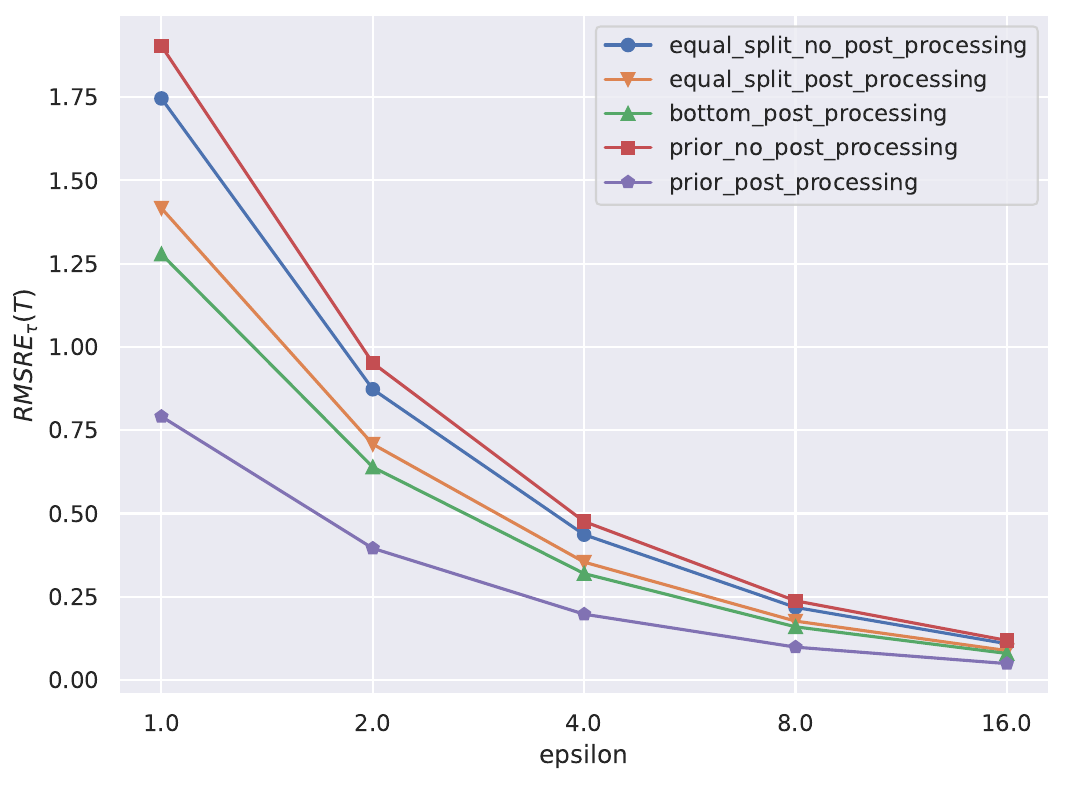}
    \label{fig:t5_noisy_ds1}
  }
  \hspace{8mm}
  \subfigure[$\tau=10$, depth $5$ tree]{
    \includegraphics[width=\myfigwidth]{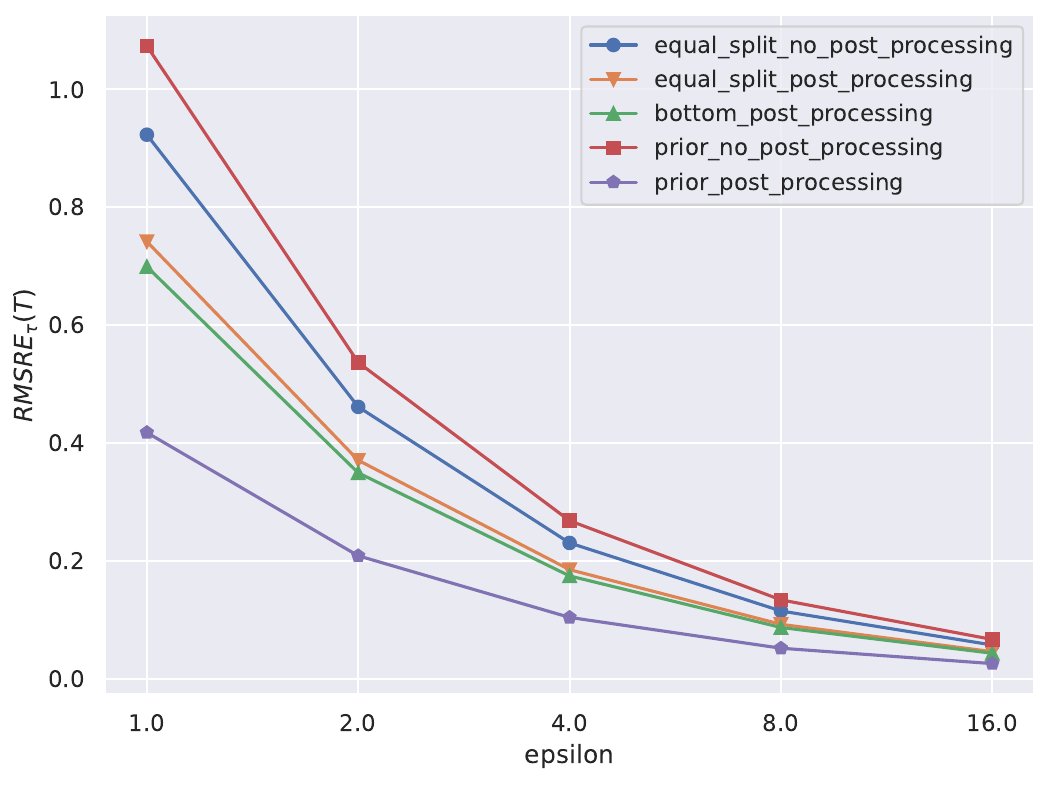}
    \label{fig:t10_noisy_ds1}
  }

  \subfigure[$\tau=5$, depth $4$ tree]{
    \includegraphics[width=\myfigwidth]{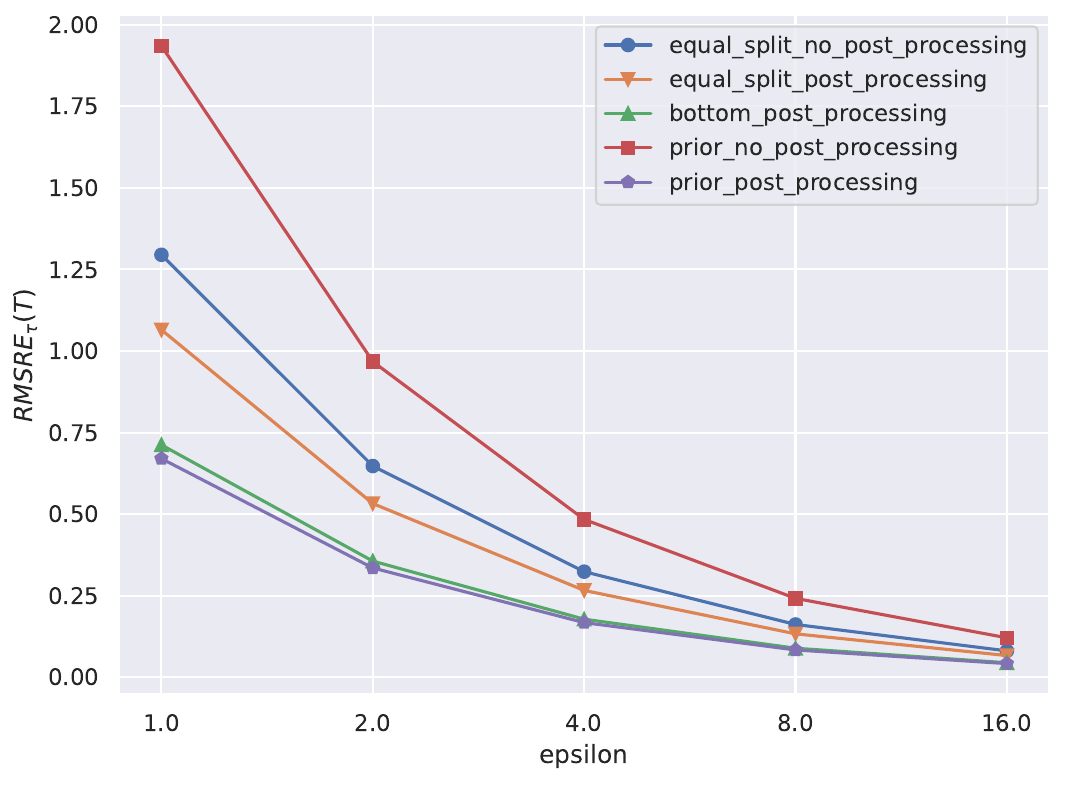}
    \label{fig:t5_noisy_simpler_tree_ds1}
  }
  \hspace{8mm}
  \subfigure[$\tau=10$, depth $4$ tree]{
    \includegraphics[width=\myfigwidth]{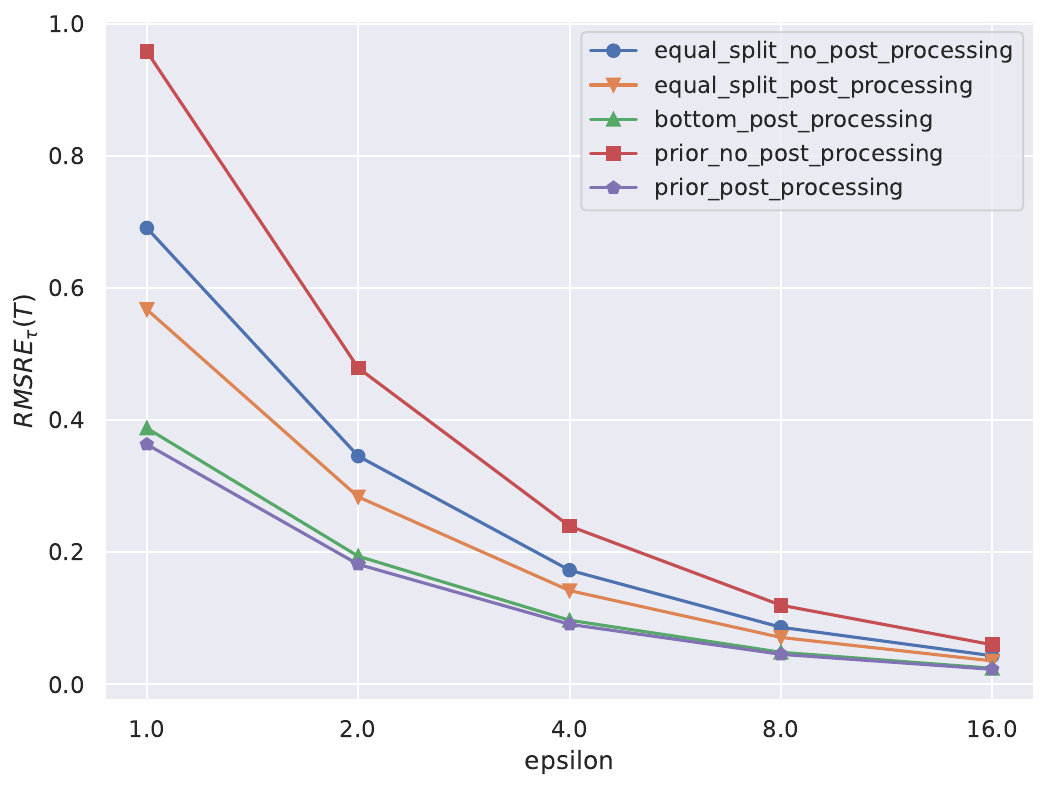}
    \label{fig:t10_noisy_simpler_tree_ds1}
  }

  \caption{$\RMSRE_\tau(T)$ vs $\eps$ for $\tau \in \{5, 10\}$ on \CSSCL{} dataset, with noisy prior obtained via equal budget split and $\eps = 1$.}
  \label{fig:rmsre_vs_eps_noisy_prior_ds1}
\end{figure*}

\begin{figure*}
  \centering
  \subfigure[$\tau=5$, depth $4$ tree]{
    \includegraphics[width=\myfigwidth]{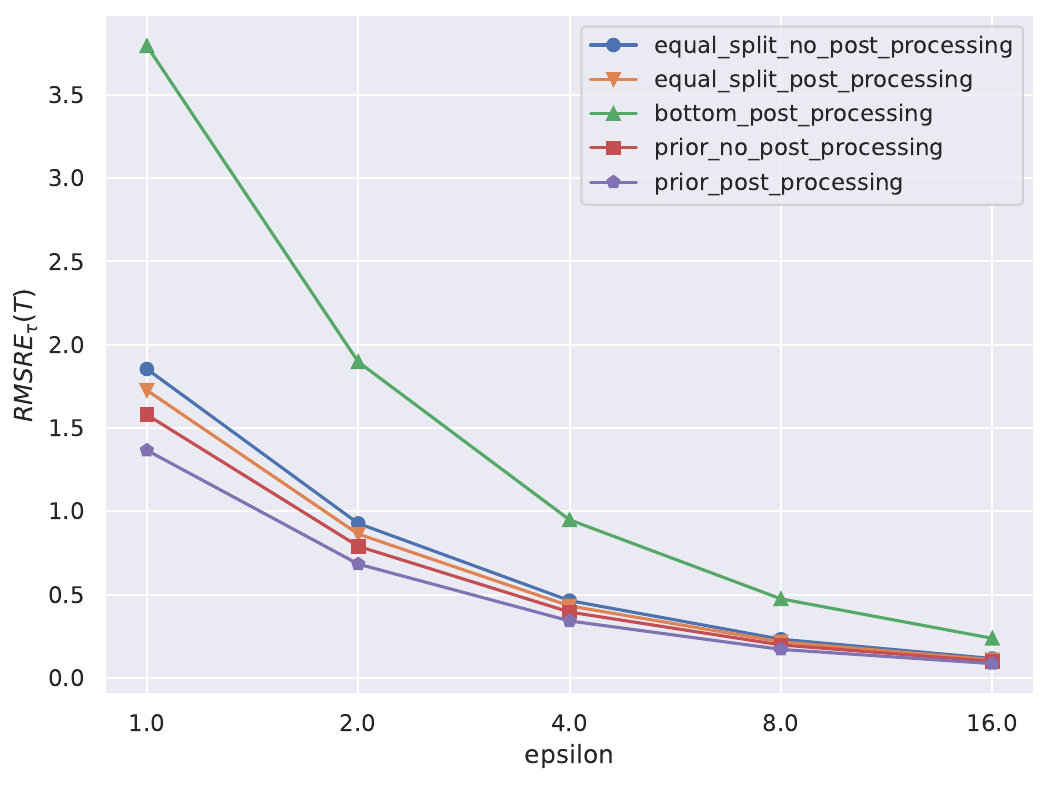}
    \label{fig:t5_noisy_ds2}
  }
  \hspace{8mm}
  \subfigure[$\tau=10$, depth $4$ tree]{
    \includegraphics[width=\myfigwidth]{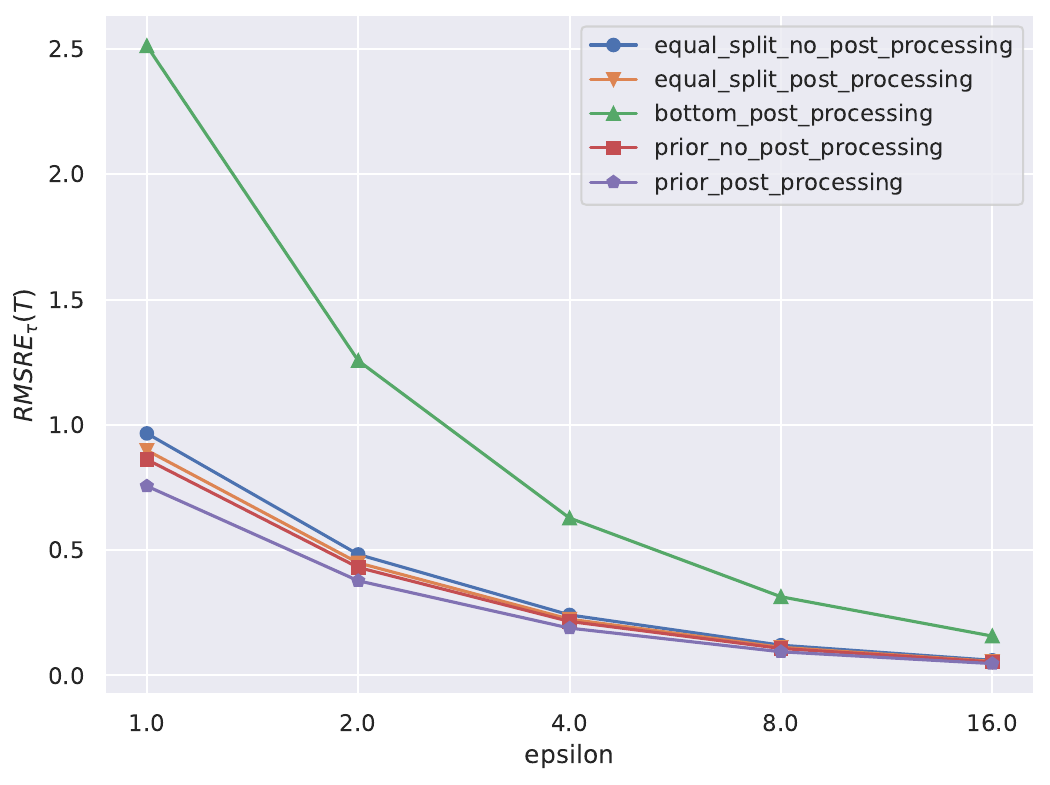}
    \label{fig:t10_noisy_ds2}
  }

  \subfigure[$\tau=5$, depth $3$ tree]{
    \includegraphics[width=\myfigwidth]{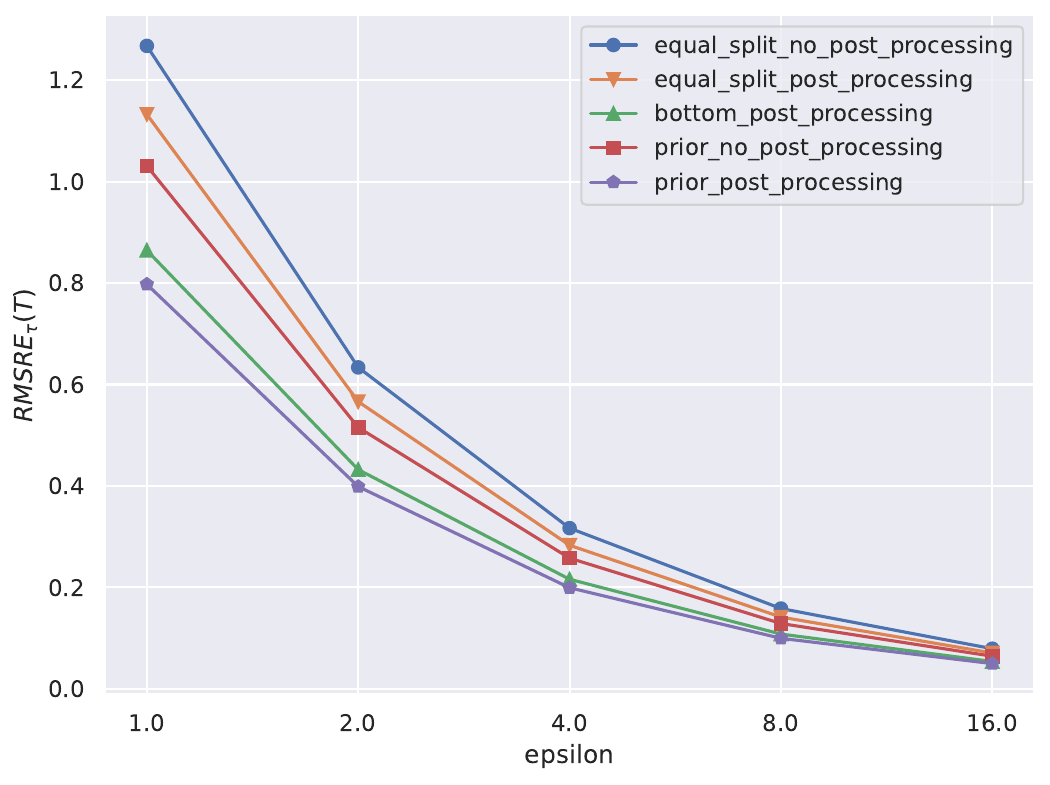}
    \label{fig:t5_noisy_simpler_tree_ds2}
  }
  \hspace{8mm}
  \subfigure[$\tau=10$, depth $3$ tree]{
    \includegraphics[width=\myfigwidth]{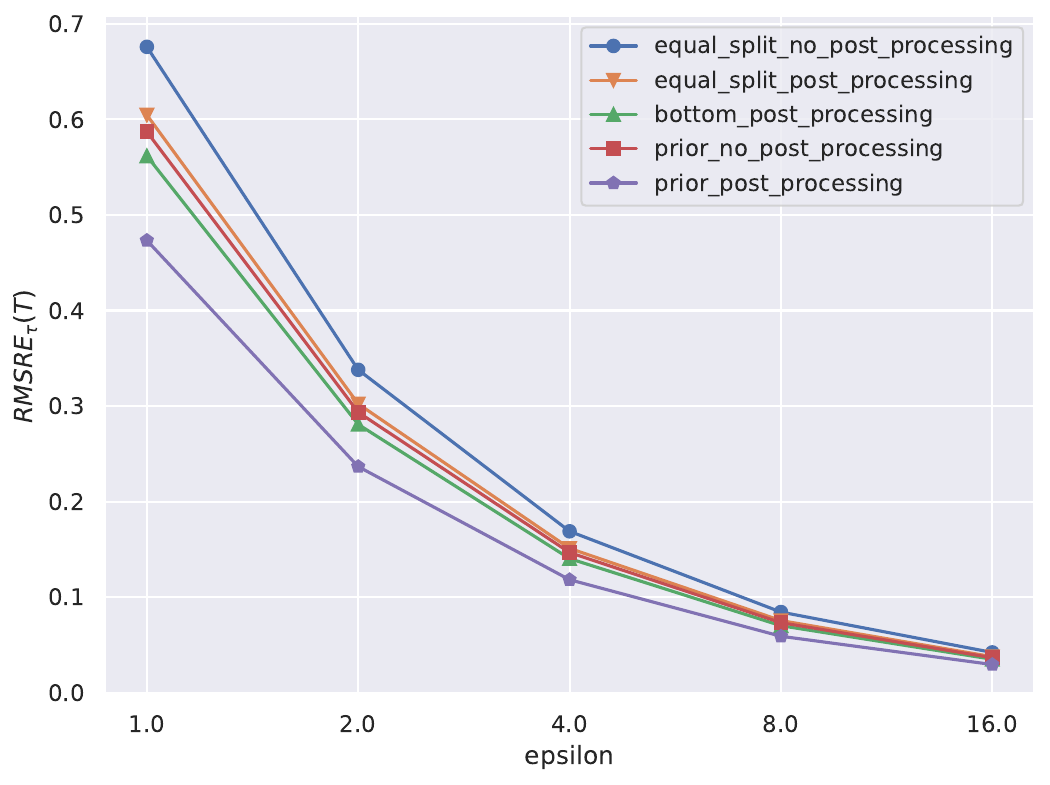}
    \label{fig:t10_noisy_simpler_tree_ds2}
  }

  \caption{$\RMSRE_\tau(T)$ vs $\eps$ for $\tau \in \{5, 10\}$ on \CAMB{} dataset with noisy prior obtained via equal budget split and $\eps = 1$.}
  \label{fig:rmsre_vs_eps_noisy_prior_ds2}
\end{figure*}

In \Cref{fig:rmsre_vs_eps_noisy_prior_ds1}, we plot $\RMSRE_\tau$ versus $\eps$ for various different methods evaluated on the later 45 days of data, namely,
\begin{itemize}[nosep,leftmargin=*]
\item Equal privacy budget split across levels without any post-processing.
\item Equal privacy budget split across levels with post-processing (\Cref{alg:post-processing}).
\item All privacy budget on leaves with post-processing.
\item Prior-based optimized privacy budget split across levels optimized for each \texttt{partner\_id} without any post-processing.
\item Prior-based optimized privacy budget split across levels optimized for each \texttt{partner\_id} with post-processing  (\Cref{alg:post-processing}).
\end{itemize}

For the prior-based privacy budget split optimization, the privacy budgeting was performed using a noisy prior computed on the first $45$ days (privately estimated with $\eps=1$ and an equal privacy budget split over levels). For \texttt{partner\_id}'s that appear only in the later $45$ days, the prior is computed on all the \texttt{partner\_id}'s that appear in the first $45$ days of data. E.g., for $\eps = 4$ and the depth $5$ tree, $\RMSRE_{10} \approx 0.1$, and for the depth $4$ tree, $\RMSRE_5 \approx 0.17$.

\paragraph{Criteo Attribution Modeling for Bidding (CAMB) Dataset \cite{criteo-attribution-modeling}}

It consists of $\sim$16M impressions from $30$ days of Criteo live traffic. We consider last-touch attribution with impression attributes: \texttt{campaign}, categorical features \texttt{cat1}, \texttt{cat8}, and a discretization of conversion delay, i.e., the gap between \texttt{conversion\_timestamp} and (impression) \texttt{timestamp}. As for \CSSCL, we consider two discretizations for the difference, with two tree depths. \Cref{fig:rmsre_vs_eps_noisy_prior_ds2} shows the plots for \CAMB{}. The privacy budgeting was performed similarly to \CSSCL{} but with the noisy prior computed on the first $15$ days. 
For $\eps = 4$ and the depth $4$ tree, $\RMSRE_{10} \approx 0.19$, and for the depth $3$ tree (omitting attribute \texttt{cat8}), $\RMSRE_5 \approx 0.20$ and $\RMSRE_{10} \approx 0.12$. %

Our prior-based budgeting with post-processing method equals or outperforms all other approaches in each setting (\Cref{fig:rmsre_vs_eps_noisy_prior_ds1,fig:rmsre_vs_eps_noisy_prior_ds2}).

\section{Conclusion and Future Directions}\label{sec:conclusion}
In this work, we studied hierarchical querying of the Attribution Reporting API, and presented algorithms for consistency enforcement and privacy budgeting, demonstrating their performance on two public ad datasets. We next discuss some interesting future directions. We focused on the so-called OPC (``one per click'') setting where each impression gets at most a single attributed conversion. An important direction is to consider the extension to the more general MPC (``many per click'') setting where an impression can get multiple attributed conversions. It would also be interesting to extend our treatment of conversion counts to the task of estimating conversion values. While the error is small for values of $\eps$ around $16$ in our evaluation, we note that this is specific to the datasets and the (restricted) functionality that we study (i.e., conversion counts with OPC). Achieving small errors on additional functionalities (e.g., MPC or conversion values) will likely require larger values of $\eps$. Another natural direction is to extend the consistency enforcement algorithm to ensure monotonicity (i.e., that the output estimate for a node of the tree is at least the output estimate for any of its children), and non-negativity. Our privacy budgeting method optimizes for one privacy parameter \emph{per-level} of the tree; it would be good to explore the extent to which \emph{per-node} privacy budgeting can yield higher accuracy. While we considered \emph{data-independent} weights when defining the tree error in terms of the node errors, there could be situations where \emph{data-dependent} weights are preferable, e.g., to avoid the tree error being dominated by the error in many nodes with no conversions, while being insensitive to the error of few nodes where most of the conversions occur. Another interesting direction to study privacy budgeting under approximate DP~\cite{dwork2006our}. Our work considered the hierarchical query model; a natural direction is to optimize the \emph{direct query model} \cite{hierarchical-vs-direct}. Finally, the Attribution Reporting API also offers \emph{event-level reports} \cite{event-api-android}. It would be interesting to see if these could be also used to further improve the accuracy for hierarchical queries.

\newpage
\mbox{}
\newpage
\mbox{}
\newpage

\balance
\bibliography{main.bbl}

\newpage
\appendix

\section{Proof of Least Squares Optimality}\label{apx:least-squares-optimality}

We provide the proof of \Cref{thm:least-squares-optimal}, starting with the proof of \Cref{lem:post-process-properties}, restated below for convenience.

\postProcessProps*

\newcommand{\zpup}[1]{{z'_{#1}}^{\uparrow}}
\newcommand{\zpUp}[1]{{z'_{#1}}^{\Uparrow}}
\newcommand{\zpdown}[1]{{z'_{#1}}^{\downarrow}}
\newcommand{\zpDown}[1]{{z'_{#1}}^{\Downarrow}}
\newcommand{\varpup}[1]{{\var'_{#1}}^{\uparrow}}
\newcommand{\varpUp}[1]{{\var'_{#1}}^{\Uparrow}}
\newcommand{\varpdown}[1]{{\var'_{#1}}^{\downarrow}}
\newcommand{\varpDown}[1]{{\var'_{#1}}^{\Downarrow}}

\begin{proof}
We will prove this by (strong) induction on the number of nodes in $T$. The base case where $T$ has a single node is immediate as the bottom-up and top-down passes are vacuous. We use the inductive hypothesis that both {\em consistency} and {\em weighted root-to-leaf sum preservation} properties hold for the output of \Cref{alg:post-processing} on any tree of at most $m$ nodes for $m \in \N$.

Let $T$ be any tree with $m + 1$ nodes. %
For any input $(x_v; \var_v)_{v \in T}$, let $(\zup_v; \varup_v)$, $(\zUp_v; \varUp_v)$, $(\zdown_v; \vardown_v)$, $(\zDown_v; \varDown_v)$, $(\hx_v; \hvar_v)$ be the values computed by \Cref{alg:post-processing} for every $v \in T$

Let $v^*$ be any internal node whose children are all leaves,
and $T'$ denote the tree with its children removed. For all $v \in T'$, define
\begin{align} \label{eq:assignment-to-reduced-tree}
(x'_v; \var'_v) =
\begin{cases}
(x_v; \var_v) & \text{ if } v \ne v^*, \\
(\zUp_v; \varUp_v) & \text{ if } v = v^*.
\end{cases}
\end{align}
Similar to above, let $(\zpup{v}; \varpup{v})$, $(\zpUp{v}; \varpUp{v})$, $(\zpdown{v}; \varpdown{v})$, $(\zpDown{v}; \varpDown{v})$, $(\hx_v'; \hvar_v')$ be as defined in \Cref{alg:post-processing} when run on the tree $T'$ with input $(x'_v; \var'_v)_{v \in T'}$. From the inductive hypothesis, for all internal nodes $v \in T'$, we have
\begin{align}\label{eq:consistent-hypothesis}
\hx'_v = \sum_{u \in \child(v)} \hx'_u,
\end{align}
 and for every leaf $v \in T'$, we have
\begin{align}\label{eq:weight-root-to-leaf-hypothesis}
\sum_{u \in \anc(v)} \frac{x'_u}{\var'_u} = \sum_{u \in \anc(v)} \frac{\hx'_u}{\var'_u}.
\end{align}

When we run \Cref{alg:post-processing} on $T$, after setting $\zUp_{v^*}$ and $\varup_{v^*}$, the rest of the bottom-up pass is exactly the same as that of the run on $T'$. Similarly, the top-down pass of $T'$ is the same as that of $T$ except that in $T$ we also set the values of $\hx_u$ and $\hvar_u$ for all $u \in \child(v^*)$. Therefore,
\begin{align} \label{eq:reduced-to-orig}
\hx_v =
\begin{cases}
\hx'_v &\text{ if } v \notin \child(v^*), \\
\frac{\vardown_v \cdot x_v + \var_v \cdot \zdown_v}{\vardown_v + \var_v} &\text{ if } v \in \child(v^*),
\end{cases}
\end{align}
where
\begin{align*}
(\zdown_v; \vardown_v) &~=~ (\textstyle \hx_{v^*} - \zup_{v^*} + x_v; \hvar_{v^*} + \varup_{v^*} - \var_v).
\end{align*}
Here we used the observation that for all leaves $v$ of $T'$ (including $v^*$), it holds that $({z'_v}^\Uparrow; {\var'_v}^\Uparrow) = (x_v; \var_v)$ and hence $({z'_v}^\Downarrow; {\var'_v}^\Downarrow) = (\hx_v'; \hvar_v') = (\hx_v; \hvar_v)$.

\paragraph{(Consistency)} For $v \ne v^*$, \Cref{eq:reduced-to-orig} implies that the LHS and RHS of the desired consistency condition is exactly the same as that of \Cref{eq:consistent-hypothesis}. So it only remains to show consistency for $v = v^*$. First, we simplify \Cref{eq:reduced-to-orig} by noting that $\vardown_v + \var_v = \varDown_{v^*} + \varup_{v^*}$ for all $v \in \child(v^*)$, and hence for $v \in \child(v^*)$ we have
\begin{align*}
\hx_v &~=~ x_v + \frac{\var_v}{\varDown_{v^*} + \varup_{v^*}} \prn{\zDown_{v^*} - \zup_{v^*}}\\
&~=~ x_v + \frac{\var_v}{\varup_v} \prn{\frac{\varup_{v^*} \zDown_{v^*} + \varDown_{v^*} \zup_{v^*}}{\varDown_{v^*} + \varup_{v^*}} - \zup_{v^*}}\\
&~=~ x_v + \frac{\var_v}{\varup_{v^*}} \prn{\hx_{v^*} - \zup_{v^*}}\,.
\end{align*}
Hence, we have
\begin{align*}
& \sum_{u \in \child(v^*)} \hx_u\\
&~=~ \sum_{u \in \child(v^*)} \left(x_u + \frac{\var_u}{\varup_{v^*}}\cdot (\hx_{v^*} - \zup_{v^*})\right)\\
&~=~\hx_{v^*}
\end{align*}
where we use that $\zup_{v^*} = \sum_{u \in \child(v^*)} x_u$. Thus, the consistency condition holds for every internal node $v \in T$ as desired.

\paragraph{(Weighted Root-to-Leaf Sum Preservation)}
Again, for $v \ne v^*$, \Cref{eq:reduced-to-orig} and \Cref{eq:weight-root-to-leaf-hypothesis} imply the weighted root-to-leaf sum preservation property for all leaves $v \notin \child(v^*)$. Meanwhile, for $v \in \child(v^*)$, we have
\begin{align}
&\sum_{u \in \anc(v)} \frac{\hx_u}{\var_u} \nonumber \\
&= \frac{\hx_{v}}{\var_v}  + \frac{\hx_{v^*}}{\var_{v^*}} - \frac{\hx_{v^*}}{\varUp_{v^*}} + \sum_{u \in \anc(v^*)} \frac{\hx'_u}{\var'_u} \nonumber \\
&\overset{\eqref{eq:weight-root-to-leaf-hypothesis}}{=} \frac{\hx_{v}}{\var_v}  + \frac{\hx_{v^*}}{\var_{v^*}} - \frac{\hx_{v^*}}{\var'_{v^*}} + \sum_{u \in \anc(v^*)} \frac{x'_u}{\var'_u} \nonumber \\
&\overset{\eqref{eq:assignment-to-reduced-tree}}{=} \frac{\hx_{v}}{\var_v}  + \frac{\hx_{v^*}}{\var_{v^*}} - \frac{\hx_{v^*}}{\varUp_{v^*}} - \frac{x_{v^*}}{\var_{v^*}} + \frac{\zUp_{v^*}}{\varUp_{v^*}} \nonumber \\
&\qquad + \sum_{u \in \anc(v^*)} \frac{x_u}{\var_u}. \label{eq:final_five}
\end{align}

\noindent Recall that
\begin{align*}
\varUp_{v^*} &~=~ \frac{\var_{v^*} \cdot \varup_{v^*}}{\var_{v^*} + \varup_{v^*}}\,,\\
\hx_v &~=~ x_v + \frac{\var_v}{\varup_{v^*}} \prn{\hx_{v^*} - \zup_{v^*}}, \\
\zUp_{v^*} &~=~ \varUp_{v^*} \cdot \left ( \frac{x_{v^*}}{\var_{v^*}} + \frac{\zup_{v^*}}{\varup_{v^*}}\right).
\end{align*}
\noindent Hence the first five terms in \Cref{eq:final_five} can be written as
\begin{align*}
&\frac{\hx_{v}}{\var_v}  + \frac{\hx_{v^*}}{\var_{v^*}} - \frac{\hx_{v^*}}{\varUp_{v^*}} - \frac{x_{v^*}}{\var_{v^*}} + \frac{\zUp_{v^*}}{\varUp_{v^*}} \\
&= \frac{1}{\var_v} \left(x_v + \frac{\var_v}{\varup_{v^*}}\cdot (\hx_{v^*} - \zup_{v^*})\right) \\
&\quad  + \frac{\hx_{v^*}}{\var_{v^*}} - \frac{\var_{v^*} + \varup_{v^*}}{\var_{v^*} \cdot \varup_{v^*}} \cdot \hx_{v^*} - \frac{x_{v^*}}{\var_{v^*}} \\
&\quad + \frac{x_{v^*}}{\var_{v^*}} + \frac{\zup_{v^*}}{\varup_{v^*}} \\
&= \frac{x_v}{\var_v}.
\end{align*}
Thus, we get that
\[
\sum_{u \in \anc(v)} \frac{\hx_u}{\var_u} ~=~ \sum_{u \in \anc(v)} \frac{x_u}{\var_u},
\]
holds for all $v \in T$. This completes our proof.
\end{proof}

We additionally need the Gauss--Markov theorem stated below for noise with non-uniform diagonal covariance.

\begin{theorem}[Gauss--Markov (see e.g., \cite{silvey1975statistical})]\label{lem:gauss-markov}
Fix $A \in \R^{n \times d}$. For an unknown $\theta \in \R^d$, suppose we observe $x = A\theta + e$ where $e \in \R^d$ is drawn such that $e_i$'s are independent with $\E e_i = 0$ and variance $\var_i$ respectively. Then, the least squares estimator $\hat{\theta}$ defined as the minimizer of $\sum_i (x_i - (A \hat{\theta})_i)^2 / \var_i$, is the best unbiased linear estimator (BLUE), namely, for all $\alpha \in \R^d$ it holds that $\Var(\alpha^T \hat{\theta})$ is the smallest among all linear unbiased estimators of $\alpha^\top \theta$.
\end{theorem}

Finally, we prove \Cref{thm:least-squares-optimal} (restated below) using \Cref{lem:post-process-properties} and \Cref{lem:gauss-markov}.

\lsqOptimality*

\begin{proof}
Corresponding to any tree $T$, we can associate the matrix $A \in \R^{n \times d}$, where $n$ is the number of nodes (both internal and leaf nodes), and $d$ is the number of leaf nodes, such that $A_{u,v}$ is $1$ if leaf $v$ is either equal to or a descendant of $u$ and $0$ otherwise. The estimated counts are given as $x = A\theta + e$ where $\theta \in \R^d$ is such that $\theta_v$ is the true count for leaf $v$. We have that $\hat{\theta}$ is a least squares estimator  minimizing $f(\tilde{\theta}) := \sum_u (x_u - (A \tilde{\theta})_u)^2 / \var_u$ if and only if it satisfies $\nabla f(\hat{\theta}) = 0$ (due to convexity of squared loss). Setting the derivative w.r.t. $\hat{\theta}_v$ equal to $0$, implies that for each leaf $v$,
\begin{align*}
    & \sum_u \frac{x_u - (A \hat{\theta})_u}{\var_u} \cdot A_{u,v} ~=~ 0\\
    \Longleftrightarrow &
    \sum_{u \in \anc(v)} \frac{x_u}{\var_u} ~=~ \sum_{u \in \anc(v)} \frac{(A \hat{\theta})_u}{\var_u}.
\end{align*}

Thus, the best linear unbiased estimate for $(A\theta)_u$ is given by $\hx_u = (A \hat{\theta})_u$, which clearly satisfies {\em consistency} and as shown above also satisfies the {\em weighted root-to-leaf sum preservation} properties of \Cref{lem:post-process-properties}. Conversely, if $(\hx_u)_u$ satisfies consistency, then it must be of the form $A \tilde{\theta}$ for $\tilde{\theta}_u = \hx_u$ for all leaves $u$, and if it satisfies the {\em weighted root-to-leaf sum preservation} property, then as shown above $\tilde{\theta}$ must be the least-squares estimator.

Thus, from \Cref{lem:post-process-properties}, we have that the $(\hx_u)_u$ returned by \Cref{alg:post-processing} satisfies the two properties, we have that each $\hx_u$ is the BLUE for the corresponding true count $c_u$ by \Cref{lem:gauss-markov}.

Finally, to see that $\hvar_v$ is indeed the variance of estimate $\hx_v$, we recursively show that $\varup_v$, $\varUp_v$, $\vardown_v$, $\varDown_v$  are the variances corresponding to $\zup_v$, $\zUp_v$, $\zdown_v$, $\zDown_v$ respectively. The key property to verify is that each $z$ quantity involves a linear combination of estimates which depend on noisy counts of disjoint parts of the tree and hence are independent. Thus, we can repeatedly use the property that for independent drawn $X$ and $Y$, it holds that $\Var(\alpha X + \beta Y) = \alpha^2 \Var(X) + \beta^2 \Var(Y)$.%
\end{proof}

\section{Greedy Iterative Budgeting}\label{apx:greedy-budgeting}

\Cref{alg:greedy-budgeting} presents the greedy iterative approach we used for optimizing the privacy budget allocation using true counts $c$, as alluded to in \Cref{sec:privacy_budgeting}. The high level idea is as follows. We choose a parameter $k$ to be a number of phases (e.g.,\ $k=20$) corresponding to the granularity of the allocation.

Initially, allocate an infinitesimal privacy budget to each level; this is done so that we can use \Cref{alg:post-processing} to compute $\hvar_v$ for each node of the tree; we need this infinitesimal allocation because \Cref{alg:post-processing} as written does not allow $\var_v$ to be $\infty$ for any $v$.\footnote{While it is possible to modify \Cref{alg:post-processing} to support $\var_v = \infty$ for certain subsets of the nodes, we avoid doing so for retaining clarity.}

We divide the (remaining) privacy budget of $\widehat\eps$ into $k$ units of size $\widehat{\eps}/k$.
In each of $k$ phases, we consider adding $\widehat{\eps}/k$ budget to all nodes at level $i$, choosing an $i$ that results in the lowest $\RMSRE_\tau(T)$, which can be computed using \Cref{alg:post-processing}; for any budgeting sequence $(\eps_0, \ldots, \eps_d)$, we set $\var_v$ to be the variance of $\DLap(1/\eps_i)$ for all nodes $v$ in level $L_i$.
Observe that $\RMSRE_\tau(T)$ can be computed directly using the variance $\hvar_v$ of each post-processed estimate $\hx_v$ as returned by \Cref{alg:post-processing} since
\[
\RMSRE_{\tau}(c_v, \hx_v) ~=~ \sqrt{ \frac{\hvar_v}{\max(\tau, c_v)^2}}\,.
\]

\begin{algorithm}
\caption{Privacy budgeting via greedy iterations}
\label{alg:greedy-budgeting}
\begin{algorithmic}
\STATE {\bf Params:} Tree $T$, with levels $L_0, L_1, \ldots, L_d$.
\STATE {\bf Input:} Total privacy budget $\eps$, Number of phases $k$
\STATE {\bf Output:} Budget split $\eps_0, \eps_1, \ldots, \eps_d$ such that $\sum_i \eps_i = \eps$.
\STATE %
\STATE Choose an infinitesimal $\gamma$, e.g., $\gamma \gets 10^{-5}$.
\FOR{$i = 0, \ldots, d$}
    \STATE $\eps_i \gets \gamma \cdot \eps / d$
\ENDFOR
\STATE $\widehat{\eps} \gets (1 - \gamma) \eps$ \hfill \textcolor{black!50}{(remaining privacy budget)}
\FOR{$j = 1, \ldots, k$}
    \FOR{$i = 0, \ldots, d$}
        \STATE $R_i \gets \RMSRE_\tau(T)$ using privacy budget split of $(\eps_0, \ldots, \eps_i + \widehat{\eps} / k, \ldots, \eps_d)$
        \STATE \textcolor{black!50}{(computed using $(\hvar_v)_{v \in T}$ from \Cref{alg:post-processing})}
    \ENDFOR
    \STATE $\ell \gets \mathrm{argmin}_i \ R_i$
    \STATE $\eps_\ell \gets \eps_\ell + \widehat{\eps} / k$
\ENDFOR
\RETURN $(\eps_0, \ldots, \eps_d)$
\end{algorithmic}
\end{algorithm}
\end{document}